\documentclass[a4paper,english,cleveref,autoref,final]{lipics-v2019}

\newif\ifproc
\procfalse

\nolinenumbers

\usepackage{xspace}
\usepackage{bm}
\usepackage{hyperref}
\usepackage{csquotes}
\usepackage{color}
\usepackage{microtype}
\usepackage{nicefrac}
\usepackage{booktabs}
\usepackage{mathtools}
\usepackage[inline]{enumitem}
\setlist[enumerate]{label={\arabic*)},font={\bfseries}}
\usepackage{etoolbox}
\usepackage{amsthm}
\BeforeBeginEnvironment{enumerate*}{%
	\setlist[enumerate,1]{label=(\alph*), font={\bfseries}}
}
\AfterEndEnvironment{enumerate*}{%
	\setlist[enumerate,1]{label=\arabic*.}
}

\usepackage[sort&compress,numbers]{natbib}

\usepackage{tikz}
\usetikzlibrary{fit,patterns,decorations.pathreplacing}

\newcommand{\Oh}{\mathcal{O}}
\newcommand{\OhOp}[1]{\Oh\mathopen{}\mathclose\bgroup\left( #1 \aftergroup\egroup\right)}

\usepackage{xspace}
\newcommand{\FPT}{{\sf FPT}\xspace}

\newcommand{\XP}{{\sf XP}\xspace}

\usepackage{tabularx}
\newcommand{\prob}[3]{
	\begin{center}
		\fbox{~\begin{minipage}{.97\textwidth}
			\vspace{2pt}
			\noindent
			\normalsize\textsc{#1}

			\vspace{4pt}
			\setlength{\tabcolsep}{3pt}
			\renewcommand{\arraystretch}{1.0}
			\begin{tabularx}{\textwidth}{@{}lX@{}}
				\normalsize\textbf{Input:}	& \normalsize#2 \\
				\normalsize\textbf{Question:}		 & \normalsize#3
			\end{tabularx}
		\end{minipage}\ \ }
	\end{center}
}

\DeclareMathOperator{\tw}{tw}
\DeclareMathOperator{\fvs}{fvs}
\DeclareMathOperator{\vc}{vc}
\DeclareMathOperator{\td}{td}

\DeclareMathOperator{\nd}{nd}

\newcommand{\CC}{\mathcal{C}}

\usepackage{stackengine}
\stackMath

\theoremstyle{theorem}
\makeatletter
\newtheorem*{rep@theorem}{\rep@title}
\newcommand{\newreptheorem}[2]{%
\newenvironment{rep#1}[1]{%
 \def\rep@title{#2 \ref{##1}}%
 \begin{rep@theorem}}%
 {\end{rep@theorem}}}
\makeatother

\newreptheorem{theorem}{Theorem}
\newreptheorem{lemma}{Lemma}
\newreptheorem{corollary}{Corollary}

\newcommand{\ubp}{\textsc{Unary Bin Packing}\xspace}

\newcommand{\sdp}{\textsc{Star Decomposition}\xspace}
\newcommand{\msdp}{\textsc{Multi-Star Decomposition}\xspace}
\newcommand{\is}{\textsc{Independent Set}\xspace}
\newcommand{\Gsa}{(G(V,E),\mathbf{s},\mathbf{a})}
\newcommand{\sa}{(\mathbf{s},\mathbf{a})}
\newcommand{\sasd}{$(\mathbf{s},\mathbf{a})$-star decomposition\xspace}

\title{Parameterized Complexity of the Star Decomposition Problem} 

\titlerunning{}

\author{Sahab Hajebi}{Department of Combinatorics and Optimization, University of Waterloo, Waterloo, Ontario, Canada}{s2hajebi@uwaterloo.ca}{}{}

\author{Ramin Javadi}{Department of Mathematical Sciences, Isfahan University of Technology, Isfahan, Iran}{rjavadi@iut.ac.ir}{https://orcid.org/0000-0003-4401-2110}{}

\authorrunning{S.\ Hajebi and R.\ Javadi}

\Copyright{Sahab Hajebi and Ramin Javadi}



\category{}

\supplement{}

\ifproc
\relatedversion{A full version of the paper is available at \url{https://arxiv.org/abs/...}.}
\fi

\ifproc\else
\hideLIPIcs  
\fi

\begin{document}

\maketitle

\begin{abstract}
A star of length $ \ell $ is defined as the complete bipartite graph $ K_{1,\ell } $.
In this paper we deal with the problem of edge decomposition of graphs into stars of varying lengths. 
 Given a graph $ G $ and  a list of integers $S=(s_1,\ldots, s_t) $, an $S$-star decomposition of $ G $ is an edge decomposition of $ G $ into graphs $G_1 ,G_2 ,\ldots,G_t $ such that $G_i$ is isomorphic to an star of length $s_i$, for each $i \in\{1,2,\ldots,t\}$.
Given a graph $G$ and a list of integers $S$, the \sdp problem asks if $G$  admits an $ S $-star decomposition. 
The problem is known to be NP-complete even when all stars are of length three. In this paper, we investigate parametrized complexity of the problem with respect to the structural parameters of the input graph such as minimum vertex cover, treewidth, tree-depth and neighborhood diversity as well as some intrinsic parameters of the problem such as the number of distinct star lengths, the maximum size of stars and the maximum degree of the input graph, giving a roughly complete picture of the parameterized complexity landscape of the problem.  
\end{abstract}



\vspace{1cm}

\section{Introduction}
By a star of length $s$ we mean the complete bipartite graph $K_{1,s}$ and the vertex of degree $s$ is called the center of the star. Given a simple graph  $G$, a decomposition of $G$ is a collection of subgraphs of $G$  say $\mathcal{G}=\{G_1,\ldots,  G_n\}$ where each edge of $G$ is in exactly one $G_i$. If all $G_i$'s are isomporphic to a star, then we say that $\mathcal{G}$ is a star decomposition.  Let $\mathcal{S}$ be a multiset of positive integers. By an $\mathcal{S}$-star decomposition, we mean a star decomposition where the lengths of stars of the decomposition are in $\mathcal{S}$ with the same multiplicities. Given a graph $G$ and a multiset $\mathcal{S}$, the \sdp problem asks if $G$ admits an $\mathcal{S}$-star decomposition. In this paper, we consider the high-multiplicity setting of the problem in which the multiset $\mathcal{S}$ is encoded with its distinct elements and its multiplicities (in binary encoding). More precisely, suppose that the distinct integers in $\mathcal{S}$ are $s_1,\ldots, s_d$ with multiplicities $a_1,\ldots, a_d$, then we assume that $\mathcal{S}$ is encoded by $(\mathbf{s},\mathbf{a})$, where $\mathbf{s}=(s_1,\ldots, s_d)$ and $\mathbf{a}=(a_1,\ldots, a_d)$. Therefore, $G$ admits an $(\mathbf{s},\mathbf{a})$-star decomposition if there exists a star decomposition $\mathcal{G}$ for $G$ such that for each $i\in [d]$, there are exactly $a_i$ stars of length $s_i$ in $\mathcal{G}$. It is clear that if such a decomposition exists, then $\mathbf{s}.\mathbf{a}=\sum_{i=1}^d a_is_i=|E(G)| $. So, the problem can be formulated as follows.

\prob{\sdp}{A graph $G=(V,E)$ and two vectors of positive integers $(\mathbf{s},\mathbf{a})$ such that $\mathbf{s}.\mathbf{a}=|E|$.}{
Does $G$ admit an $(\mathbf{s},\mathbf{a})$-star decomposition?}

The \sdp problem can be seen as a two-stage problem where in first stage, we choose an orientation on the edges of the graph $G$ and in second stage we distribute the star lengths on vertices of $G$ such that the sum of  star lengths assigned to each vertex $v$ is equal to the out-degree of $v$. The second stage is in fact a bin packing problem in which a set of items with prescribed weights (star lengths) as well as a set of bins (vertices of $G$) with different capacities (out-degrees) are given and the task is to assign items to bins such that the sum of item weights assigned to each bin does not exceed its capacity. 

In this paper, we will investigate parameterized complexity of \sdp by considering structural parameters of the input graph such as treewidth ($\tw$), tree-depth ($\td$), neighborhood diversity ($\nd$) and vertex cover number ($\vc$) (for the definitions, see e.g. \cite{lampis}) as well as the intrinsic parameters of the problem such as $d$ (the number of star types), $s=\max_i s_i$ (the maximum length of stars) and $\Delta$ (the maximum degree of the graph). It can be easily seen that 
\begin{align}
d&\leq s\leq \Delta,\\
\tw&\leq \td \leq \vc,\label{eq:tw}\\
\nd&\leq 2^{\vc}+\vc. \label{eq:nd}
\end{align}
For the proof of Inequalities \eqref{eq:tw} and \eqref{eq:nd} see \cite{lampis}. It is known that the decomposition of a graph into stars of lengths three is NP-hard on graphs of maximum degree $4$ (see \cite{dyer}). Therefore, the problem is para-NP-hard with all parameters $d,s,\Delta$. 
\begin{theorem}{\rm \cite{dyer}}\label{thm:NPhard}
The problem of decomposition of a graph of maximum degree four into stars of length three $K_{1,3}$ is NP-hard. Consequently, \sdp is NP-hard for $d=1$, $s=3$ and $\Delta=4$. 
\end{theorem}

We begin by improving the above result and proving that \sdp is NP-hard even on Hamiltonian planar graphs when  $d=2$, $s=3$ and $\Delta =3$. We also prove that this is the border line for NP-hardness, i.e. the problem can be solved in polynomial time on graphs of maximum degree $3$ whenever $s\leq 2$ or $(s,d)=(3,1)$. 

By considering \sdp on trees and complete graphs, we will show that \sdp is W[1]-hard with respect to vertex cover number ($\vc$) on the class of trees of depth at most two and complete graphs (see Theorems~\ref{thm:W1} and \ref{TreesW1}). Therefore, the problem is para-NP-hard with respect to treewidth ($\tw$), tree-depth ($\td$) and neighborhood diversity ($\nd$). Hence, we will consider the combined parameters as follows. We will prove that the problem is W[1]-hard with respect to $(\td,d)$ and $(\nd,d)$ and is FPT with respect to $(\vc,d)$, $(\td,s)$ and $(\nd,s)$. We will also prove that \sdp is in XP with respect to $\vc$ as well as $(\tw,d)$. A summary of our results is shown in Table~\ref{tbl:results}.

\begin{table}[t]
	\centering\footnotesize
	\caption{\label{tbl:results}Overview of new NP-hardness and parameterized results for the \sdp. $\mathscr{T}$ and $\mathscr{CB}$ stands for the class of trees and complete bipartite graphs, respectively.}
	\vskip 7pt
	\begin{tabular}{ p{.03\textwidth}  p{.2\textwidth}  p{.25\textwidth}  p{.15\textwidth}  p{.2\textwidth} }     
		\toprule
		 & & d & s & $\Delta$  \\%
				\cmidrule{3-5}
				&& NP-h for $d=1$    & NP-h for $s=3$ & NP-h for $\Delta=3$ \\
				&&  [Th. \ref{thm:NPhard}]   &  [Ths. \ref{thm:NPhard},\ref{delta=3}]&  [Th. \ref{delta=3}]  \\
    
				&& W[1]-h for $\mathscr{CB}$ and $\mathscr{T}$ & FPT for   $\mathscr{CB}$   & 
                \\
   & & [Ths. \ref{thm:W1},\ref{TreesW1}]  & 
   & \\
		\cmidrule{2-5}
		$\tw$& NP-h for $\tw=1$  &
		W[1]-h [Th. \ref{TreesW1}]&  
        &\\
		 & [Th. \ref{TreesW1}] & XP [Cor. \ref{cor:d,tw}] & XP  & XP \\
   
		\cmidrule{2-5}
		$\td$& NP-h for $\td=2$  &
		 W[1]-h for $\mathscr{T}$ [Cor. \ref{cor:td,d}] & FPT [Cor. \ref{cor:s,td}]& FPT \\
		 & [Th. \ref{TreesW1}] & XP  \\
		 		 \cmidrule{2-5}
$\nd$& NP-h for $\nd=2$&
W[1]-h [Cor. \ref{cor:vc,d}] & FPT [Th.~\ref{thm:s,nd}] & FPT \\
&  [Cor. \ref{cor:vc,d}]   & 
&  \\
		 \cmidrule{2-5}
		 $\vc$& XP [Th. \ref{XPvc}] &
		FPT  [Th. \ref{FPT:vc,d}]& FPT & FPT \\
		 & W[1]-h for $\mathscr{CB}$ and $\mathscr{T}$ \\
   & [Ths. \ref{thm:W1},\ref{TreesW1}] & & \\
		\bottomrule
	\end{tabular}
\end{table}	
\subsection{Related work}
The decomposition problems on graphs have been extensively studied in the literature. In particular, the study of decomposition of a graph into stars dates back to 1975 in which Yamamoto et al. \cite{yamam75} found necessary and sufficient conditions for the decomposition of a complete graph or complete bipartite graph into stars of length $\ell$. This result was then generalized to the decomposition of complete multipartite graphs \cite{yamam78} and complete multigraphs \cite{tarsi79}. Also, in 1980, Tarsi \cite{tarsi} found a sufficient condition for the existence of a decomposition of a graph into stars of varying sizes (see Theorem~\ref{thm:tarsi}). In particular, he proved that if $s=\max_i s_i\leq n/2$, then $K_n$ admits an $(\mathbf{s},\mathbf{a})$-star decomposition (provided that the obvious necessary condition $\mathbf{s}.\mathbf{a}=n(n-1)/2$ is satisfied).

The main conjecture about the complexity of decomposition problems was Holyer's conjecture \cite{holyer} which asserts that given a graph $H$, the problem of existence of an $H$-decomposition for a given graph is NP-complete whenever $H$ has a connected component with at least three edges. After a series of partial results, Holyer's conjecture was completely settled by Dor and Tarsi in a seminal work \cite{dor}.
A recent result is due to Cameron et al. \cite{cameron} which studied NP-completeness of the decomposition problem of complete multigraph into stars of varying sizes. 

 A related problem is studied in \cite{cicalese} in which we are aimed to decompose the graph into stars such that the minimum length of stars in the decomposition is as large as possible. They proved that the problem is NP-hard for planar graphs with maximum degree four, and for a given graph, it is hard to check whether there exists a star decomposition with stars of length at least three (it appears to be tractable when three is replaced with two). They also proved that the problem is polynomial-time solvable on trees as well as graphs with maximum degree three, and they provide a linear time 1/2-approximation algorithm.

Another related problem is \textsc{Star Partition} in which, given a graph $G$ and positive integer $k$, one seeks for a partition of the vertex set (instead of the edge set) of $G$  into $k$ sets $V_1,\cdots,V_k$ such that the induced graph of $G$ on $V_i$ is a star. The problem is known to be NP-complete for many classes including chordal bipartite graphs \cite{muller}, split graphs, line graphs, subcubic bipartite planar graphs \cite{shalu}, and $(C_4 
,\cdots ,C_{2t})$-free bipartite graphs for every fixed
$t \geq 2$ \cite{duginov}. On the other hand, it is known that the problem is polynomially solvable for bipartite permutation graphs \cite{brandst, farber}, convex bipartite graphs \cite{bang, damaschke}, doubly convex bipartite graphs \cite{bang}, and trees \cite{hedetniem}. Form the parameterized complexity aspect, the problem is W[2]-hard  with respect to $k$ for bipartite graphs \cite{raman}. However, the problem is fixed parameter tractable with respect to $k$  on split graphs and graphs of girth at least five \cite{divya, raman}, and with respect to parameters vertex cover and treewidth on general graphs \cite{Nguyen}.

\subsection{Notations and Terminology}
For two integers $m,n$, the sets $\{1,\ldots, n\}$ and $\{m,\ldots, n\}$ are denoted by $[n]$ and $[m,n]$, respectively. For a (multi-)set $A$ of integers, we define $\sigma(A)=\sum_{a\in A} a$. For a graph $G=(V,E)$ and two subsets of vertices $A,B\subseteq V$, the set of all edges with one endpoint in $A$ and one endpoint in $B$ is denoted by $E_G(A,B)$ and we define $d_G(A,B)=|E_G(A,B)|$ and $E_G(A)=E_G(A,A)$ (we drop subscripts whenever there is no ambiguity). Also, we say that $A$ is complete (resp. incomplete) to $B$ if every vertex in $A$ is adjacent (resp. nonadjacent) to every vertex in $B$. For a subset $S\subseteq V$, the induced graph of $G$ on $S$ is denoted by $G[S]$. Let $G=(V,E)$ be a digraph. For every vertex $v\in V$, $N_G^+(v)$ (resp. $N_G^-(v)$) stands for the set of vertices $u$ where $vu$ (resp. $uv$) is an edge in $G$. Also, define $\deg^+_G(v)=|N^+_G(v)|$ and $\deg^-_G(v)=|N^-_G(v)|$. Moreover, for a subset $S\subseteq V$, we define $N^+_S(v)=N^+_G(v)\cap S$ and $N^-_S(v)=N^-_G(v)\cap S$.


\section{Tools}
In this section, we gather all necessary tools that we need through the paper. 
One of the main tools that we deploy is formulating \sdp in different integer linear programming (ILP) formats and then applying existing machinery on solving integer programmings. Consider a general standard form of an ILP as follows.
\begin{equation} \label{eq:ILP}
\min\{\mathbf{w}.\mathbf{x} \mid A\mathbf{x} = \mathbf{b},\ \mathbf{l} \leq \mathbf{x} \leq \mathbf{u},\ \mathbf{x} \in \mathbb{Z}^n\},
\end{equation}
where $\mathbf{w}.\mathbf{x}$ is the dot product of $\mathbf{w}$ and $\mathbf{x}$.
It is well-known that ILP is NP-hard in general \cite{cygan}. The most classic result regarding fixed-parameter tractability of ILP is an old result due to Lenstra which states that ILP is in FPT with respect to the number of variables.

\begin{theorem} {\rm \cite{lenstra}} \label{thm:ILPvar}
An ILP feasibility instance of size $L$ with $p$ variables can be solved using $O(p^{2.5 p+o(p)}.L)$ arithmetic operations and space polynomial in $L$.
\end{theorem}

Later, many other techniques beyond Lenstra's algorithm developed for fixed-parameter solvability of ILP. One of the nice approach is solving an ILP when the coefficient matrix is an special matrix such $N$-fold, tree-fold, etc. In this regards, we can assign some invariants to the coefficient matrix $A$ and solve the ILP in FPT with respect to these invariants. 

The \textit{dual graph} corresponding to the matrix $A$ denoted by $G_D(A)$ is a graph which has a vertex for each row of $A$ and two vertices are adjacent if there exists a column of $A$ such that both corresponding rows are non-zero. The \textit{tree-depth} of a graph $G$ denoted by $\td(G)$ is defined as the minimum height of a rooted forest $F$ such that $V(F)=V(G)$ and for each edge $uv \in E(G)$, $u$ is either an ancestor or a descendant of $v$ in $F$. The \textit{dual tree-depth} of $A$ is defined as $\td_D(A)=\td(G_D(A))$. 
Also, the \textit{dual treewidth} of $A$ is defined as $\tw_D(A)=\tw(G_D(A))$. 

We need the following result which states that an ILP can be solved in FPT time with respect to $\|A\|_{\infty}$ and $\td_D(A)$. 

\begin{theorem}{\rm \cite{tecrep}} \label{thm:tdILP}
	The general ILP problem \eqref{eq:ILP} can be solved in time $O^*((\|A\|_{\infty} + 1)^{2^{\td_D(A)}})$. 
\end{theorem}

It is known that the above result cannot be generalized by replacement of $\td_D(A)$ with $\tw_D(A)$ since ILP is NP-hard even when $\|A\|_\infty=2$ and $\tw_D(A)=2$ \cite{tecrep}. However, Ganian et al. \cite{ganian} proved that
\begin{theorem} {\rm \cite{ganian}} \label{thm:twD}
	The general ILP problem \eqref{eq:ILP} can be solved in time $\Gamma^{O(\tw_D(A))} .n$, where 
\begin{equation} \label{eq:Gamma}
\Gamma= \max_{\mathbf{x}\in \mathbb{Z}^n: A\mathbf{x}=\mathbf{b},\ \mathbf{l}\leq \mathbf{x}\leq \mathbf{u}} \ \max_{k\in [n]} \left\|\sum_{j=1}^k A_jx_j\right\|_\infty. 
	\end{equation}
\end{theorem}

We also need a generalization of Hall's Theorem regarding existence of an SDR for a set family. Let $\mathcal{F}=(F_i: i\in I)$ be a family of sets and $\eta: I\to \mathbb{Z}^+$ be a function. An $\eta$-SDR for $\mathcal{F}$ is a family $(S_i: i\in I)$ of disjoint subsets such that $S_i\subseteq F_i$ and $|S_i|=\eta(i)$, for each $i\in I$. The following theorem gives a necessary and sufficient condition for the existence of an $\eta$-SDR for $\mathcal{F}$.

\begin{theorem} \label{thm:SDR}
Let  $\mathcal{F}=(F_i: i\in I)$ be a family of sets and $\eta: I\to \mathbb{Z}^+$ be a function. There exists an $\eta$-SDR for $\mathcal{F}$ if and only if 
\[
|\bigcup_{j\in J} F_j|\geq \sum_{j\in J} \eta(j),\ \forall J\subseteq I,
\]
and if the above conditions hold, then an $\eta$-SDR for $\mathcal{F}$ can be found in $O(n^3)$, where $n=|\cup_{i\in I} F_i|$.
\end{theorem}

Finally, we need the following result from \cite{tarsi} which guarantees the existence of a star decomposition in good expanders. 

Given a graph $G=(V,E)$, the \textit{edge expansion} of $G$ is defined as 
\begin{equation} \label{eq:expansion}
    \varphi(G)= \min_{S\subsetneq V: S\neq \emptyset } \dfrac{1}{2}\left(\dfrac{1}{|S|}+\dfrac{1}{|V\setminus S|}\right)d(S,V\setminus S).
\end{equation}

\begin{theorem}\label{thm:tarsi}
Let $I=\Gsa$ be an instance of \sdp. If $s=\max_i{s_i} \leq \varphi(G)$, then $I$ is a yes instance.    
\end{theorem} 

In Section~\ref{sec:nd,s}, we will generalize the above result and prove that if $C$ is a vertex cover of $G$ and $s\leq \varphi(G[C])$, then $G$ admits an $(\mathbf{s},\mathbf{a})$-star decomposition. 

\section{Hardness results}
In this section, we prove that \sdp is NP-hard when $(d,s,\Delta)=(2,3,3)$ and this is the border line for NP-hardness of the problem in the sense that it is polynomially solvable when either $s\leq 2$ or $(d,s,\Delta)=(1,3,3)$ (by Theorem~\ref{thm:NPhard}, we know that \sdp is NP-hard when $(d,s,\Delta)=(1,3,4)$). Also, we prove that the problem is W[1]-hard with respect to the parameters  $\vc$ and $d$ on the class of all complete bipartite graphs as well as the class of trees with depth $2$. 

\begin{theorem}\label{delta=3}
\sdp is polynomially solvable when $s\leq 2$. It is also polynomially solvable on graphs of maximum degree three when $s=3$ and $d=1$. However, it is NP-hard on Hamiltonian planar cubic graphs when $s=3$ and $d=2$.
\end{theorem}
\begin{proof}
Let $\Gsa$ be an instance of \sdp. If $s\leq 2$, we can prove that $G$ has an \sasd if and only if the number of connected components of $G$ with odd sizes is at most $a_1$. To see this, note that in any \sasd of $G$, any connected component with odd number of edges has at least one star of length one, therefore the number of odd connected components is at most $a_1$. For the reverse, we can remove one edge from each odd connected component to ensure that all connected components have an even number of edges. Then, it is known that every connected graph with even number of edges can be decomposed into paths of length two. So, the star decomposition exists.    

Now, suppose that $(d,s,\Delta)=(1,3,3)$, i.e. we want to decompose the graph $G$ into stars of length three. If such decomposition exists, then the centers of stars is a vertex cover which is also an independent set. Therefore, $G$ is a bipartite graph where all vertices in one part are of degree three. It is clear that this condition is also sufficient for the existence of the decomposition of $G$ into stars of length three. The condition can be evidently checked in polynomial time. 

Finally, we prove that \sdp is NP-hard on Hamiltonian planar cubic graphs when $s=3$ and $d=2$. We give a reduction from \is on Hamiltonian planar cubic graphs which is known to be NP-hard \cite{herbert}. Given a Hamiltonian planar cubic graph $G$ and integer $k$ as an instance of \is, take the same graph $G$ and let $\mathbf{s}=(1,3)$ and $\mathbf{a}=(3n/2-3k,k)$. 
We show that $G$ has an independent set of size $k$ if and only if $G$ admits an $(\mathbf{s},\mathbf{a})$-star decomposition. First, let $G$ has an independent set $I=\{v_1,\ldots,v_k\}$. So, for each $v_i$ consider a star of length $3$ with center $v_i$ and cover all remaining edges with stars of length one. This gives an $(\mathbf{s},\mathbf{a})$-star decomposition for $G$. Now, suppose that $G$ admits an $(\mathbf{s},\mathbf{a})$-star decomposition and $I=\{v_1,\ldots,v_k\}$ is the set of centers of stars of length $3$ in this decomposition. For each $i,j\in [k]$, $v_i$ is not adjacent to $v_j$, since they have degree three and they cannot share any edge. Thus, $I$ is an independent set of size $k$ of $G$.
\end{proof}

\begin{theorem}\label{thm:W1} 
\sdp is W[1]-hard on the class of all complete bipartite graphs when parameterized by the parameters $\vc$ or $d$.  
\end{theorem}
\begin{proof} 
We give a parameterized reduction from \textsc{Unary Bin Packing},
	which is defined as follows
	
\prob{\ubp}{
 Given $ d $ item types $w_1,\ldots, w_d\in \mathbb{Z}^+$ in increasing order and for each $i\in [d]$, there are $a_i$ items of weight $w_i$. Also, integers $ m,B \in {\mathbb{Z}^{+}}$ are given. The input $(\mathbf{w},\mathbf{a},m,B)$ are given in unary encoding. 
}{Can we partition the items into  $ m $ bins such that the sum of the weights of items in each bin is at most $ B $?}

	Note that if $\sum_{i=1}^da_i{w_i}> Bm$, then the answer is obviously no. So, without loss of generality, we can suppose that $\sum_{i=1}^da_i{w_i}\leq Bm$. Also, by adding $Bm-\sum_{i=1}^da_i{w_i}$ items of weight one, we can assume that $\sum_{i=1}^da_i{w_i}= Bm$ and in any solution, the sum of the weights of items in each bin is equal to $B$. 
	
	 It is proved that \ubp is strongly NP-complete \cite{Garey} and also W[1]-complete when parameterized by the number of bins $m$  \cite{Jansen} and also the number of types $d$ \cite{koutecky}. 
	
	Suppose that  $(\mathbf{w},\mathbf{a},m,B)$ is an instance of \ubp, where $\sum_{i=1}^da_i{w_i}= Bm$. Set $ n=B(m+1)$ and consider the complete bipartite graph $K_{m,n}$ with bipartition $(X,Y)$, where $X=\{v_1,\ldots, v_m\}$ and $Y=\{u_1,\ldots,u_n \}$. 
	 Now, for each $  i\in [d] $, define $ s_i= (m+1) w_i $. We claim that $ K_{m,n} $ admits an $ (\mathbf{s},\mathbf{a}) $-star decomposition if and only if $(\mathbf{w},\mathbf{a},m,B)$ is a yes-instance for \ubp.
		
 First, suppose that the items can be partitioned into  $ m $ bins (multisets) $ A_1,A_2,\ldots,A_m $ such that for each $ i\in [m] $, $ \sigma(A_i) = B $.  
	Now, for each $i\in [d]$ and for each item of weight $w$ in $A_i$, consider a star of length $(m+1)w$ with the center $v_i$. This gives an $(\mathbf{s},\mathbf{a})$-star decomposition for $K_{m,n}$. 
	
	Now, suppose that  there exists an $(\mathbf{s},\mathbf{a})$-star decomposition $\mathcal{D}$ for $K_{m,n}$. 
	Since all $s_i$'s are greater than $m$, the centers of all stars in $\mathcal{D}$ are in $X$. 
	Now, for each $ i\in[m]$,  let $ M_i $ be the multiset of lengths of all stars in $\mathcal{D}$ with center $ v_i $. It is clear that $\sigma(M_i)= (m+1)B$. Then, the multisets $ A_i=(m+1)^{-1}M_i $, $i\in [m]$, pack all items into $m$ bins of size $B$. Hence, $(\mathbf{w},\mathbf{a},m,B)$ is a yes-instance for \ubp.
	
	Note that since the input is given in unary encoding, this is a polynomial-time reduction. Also, since $\vc=m$ and \ubp is W[1]-hard with respect to both parameters $d$ and $m$, \sdp is W[1]-hard with respect to both parameters $d$ and $\vc$ on the complete bipartite graphs.
\end{proof}

Since the complete bipartite graph has neighborhood diversity equal to two, we have the following corollary.
\begin{corollary}\label{cor:vc,d}
\sdp is W[1]-hard on the class of all graphs of neighborhood diversity at most two when parameterized by the parameters $\vc$ or $d$.	
\end{corollary}

In the following theorem, we prove W[1]-hardness of \sdp with respect to the parameters $\vc$ or $d$ on another class of graphs which is trees of depth at most two.

\begin{theorem}\label{TreesW1}
\sdp is W[1]-hard on the class of all trees of depth at most two when parameterized by the parameters $\vc$ or $d$.
\end{theorem}

\begin{proof}
We again provide a parameterized reduction from \ubp. Let $(\mathbf{w},\mathbf{a},m,B)$ be an instance of \ubp. Again, without loss of generality, we can assume that $\sum_{i=1}^d a_iw_i=Bm$ and for each $i\in [d]$, $w_i\leq B$. Let $\ell =\max\{m,B+2\}$. Define the rooted tree $T$ with a root $r$ of degree $\ell$ with children $u_1,\ldots, u_{\ell}$ such that for each $i\in [m]$, $u_i$ has exactly $B$ children and for each $i\in [m+1,\ell]$, $u_i$ is a leaf. 
Also, for each $i\in [d]$, define $s_i=w_i$, $\hat{a}_{i}=a_i$ and $s_{d+1}= \ell$ and $\hat{a}_{d+1}=1$.
We claim that $(\mathbf{w},\mathbf{a},m,B)$ is a yes-instance of \ubp if and only if $T$ admits an $(\mathbf{s},\mathbf{\hat{a}})$-star decomposition. 

First, let $(A_1,\ldots, A_m)$ be a partition of items into $m$ bins such that $\sigma(A_i)=B$. 
	Now, for each $i\in [d]$ and for each item of weight $w$ in $A_i$, consider a star of length $w$ with the center $u_i$. Also, consider a star of length $\ell$ with the center $r$. This gives an $(\mathbf{s},\mathbf{\hat{a}})$-star decomposition for $T$. Conversely, let $\mathcal{D}$ be an $(\mathbf{s},\mathbf{\hat{a}})$-star decomposition of $T$. Since $r$ is the only vertex of $T$ with degree $\ell$, the center of the unique star of length $s_{d+1}=\ell$ is the vertex $r$. Now, for each $i\in [m]$, define $A_i$ be the multiset of lengths of all stars in $\mathcal{D}$ with center $u_i$. Since each $u_i$ has exactly $B$ children, it is clear that $\sigma(A_i)=B$. This implies that $(\mathbf{w},\mathbf{a},m,B)$ is a yes-instance for \ubp. 
	
	Note that the set $\{r,u_1,\ldots, u_m\}$ is a vertex cover of $T$ and so $\vc(T)\leq m+1$. Also, $T$ is a tree of depth two.  Since \ubp is W[1]-hard with respect to both parameters $m$ and $d$ (\cite{Jansen},\cite{koutecky}), \sdp is also $W[1]$-hard with respect to both parameters $\vc$ and $d$ on the class of all trees of depth at most two. 
\end{proof}

\begin{corollary}\label{cor:td,d}
	\sdp is W[1]-hard with respect to the combined parameter $(\td,d)$ even on the class of trees.	
\end{corollary}

\section{ILP formulations}
Let $I=\Gsa$ be an instance of \sdp. It can be seen that $I$ is a yes-instance if and only if there is an orientation $\tau$ on $G$ such that the oriented $G_\tau$ admits an $\sa$ star decomposition where each star is directed from its center to its leaves. Now, we can write an integer linear programming (ILP) for the problem as follows. Suppose that $V=[n]$ and for each vertex $u\in V$ and star length $i\in [d]$, define the variable $x_{i,u}$ as the number of stars of length $s_i$ with the center $u$. Also, for each edge $e=uv\in E$, $u<v$, define the variable $y_e$ which is equal to one if $e$ is oriented in $\tau$ from $u$ to $v$ and is equal to zero, otherwise. Now, it is clear that $I$ is a yes-instance if and only if the following ILP has a feasible solution. 

\begin{align} {\rm ILP1:} & \nonumber\\
&\sum_{u\in V} x_{i,u} = a_i & \forall\, i \in [d] \label{eq:d}\\
&\sum_{i = 1}^{d} s_i x_{i,u} = \sum_{\stackrel{uv\in E}{u<v}} y_{uv}+\sum_{\stackrel{uv\in E}{v<u}} (1-y_{vu})  & \forall\, u \in V \label{eq:V}\\
&x_{i,u} \in \mathbb{Z}^{\geq 0} &\forall\, i \in [d] \text{ and } \forall u\in V \nonumber\\
&y_{uv} \in \{0,1\}  & \forall\, uv\in E \nonumber
\end{align}
\vspace{2mm}

Constraints \eqref{eq:d} assure that the total number of used stars of length $s_i$ on all vertices is equal to $a_i$, for each $i\in [d]$ and Constraints \eqref{eq:V} guarantee that the sum of lengths of stars on vertex $u$ is equal to the out-degree of the vertex $u$, for each $u\in V$. 

Now, we compute the dual tree-depth of ILP1 as follows. Let $A$ be the coefficient matrix of ILP1 which has $n+d$ rows corresponding to $n$ constraints in \eqref{eq:d} and $d$ constraints in \eqref{eq:V}. So, the vertex set of the graph $G_D(A)$ is $U\cup W$ such that $U=\{u_1,\ldots, u_n\}$ corresponds to the constraints in \eqref{eq:V} and $W= \{w_1,\ldots w_d\}$ corresponds to the constraints in \eqref{eq:d}. Note that since the constraints in \eqref{eq:d} have no common variable, $G_D(A)$ induces a stable set on $W$. Also, for two adjacent vertices $u$ and $v$ in $G$, their corresponding constraints in \eqref{eq:V} have the variable $y_{uv}$ in common, so $G_D(A)$ induces the graph $G$ on $U$. Finally every constraint in \eqref{eq:d} have a common variable $x_{i,u}$ with every constraint in \eqref{eq:V}. Therefore, $U$ is complete to $W$ in $G_D(A)$. Hence, $G_D(A)$ is the join of the graph $G$ and an stable set of size $d$ and thus, $\td_D(A)\leq \td(G)+d$ and $\tw_D(A)\leq \tw(G)+d$. Also, $\|A\|_\infty=s=\max_i s_i$. Now, we can apply Theorem~\ref{thm:tdILP} to prove that \sdp is FPT with respect to $(\td(G),s)$ (note that $d\leq s$).\\

\begin{corollary} \label{cor:s,td}
Let $\Gsa$ be an instance of \sdp, $\td$ be the tree-depth of the graph $G$ and $s=\max_i s_i$. The problem can be solved in time $O^*((s + 1)^{2^{\td+d}})$. In particular, \sdp is in FPT with respect to $(\td,s)$.
\end{corollary}

Moreover, we can apply Theorem~\ref{thm:twD} to prove that \sdp is in XP with respect to $(\tw,d)$.\\

\begin{corollary} \label{cor:d,tw}
	Let $\Gsa$ be an instance of \sdp, $\tw$ and $\Delta$ be the treewith and the maximum degree of the graph $G$, respectively, $a=\max_i a_i$ and $d$ be the number of star lengths. The problem can be solved in time $O^*(\max({a,\Delta})^{O(\tw+d)})$. In particular, \sdp is in XP with respect to $(\tw,d)$.
\end{corollary}

\begin{proof}
In the light of Theorem~\ref{thm:twD} and the fact that in ILP1, we have $\tw_D(A)\leq \tw(G)+d$, it is just enough to prove that in ILP1, $\Gamma\leq \max({a},\Delta)$ ($\Gamma$ is defined in \eqref{eq:Gamma}). For this, let $A$ be the coefficient matrix of ILP1 and let $\mathbf{z}=(\mathbf{x},\mathbf{y})$ be a feasible solution for ILP1. For each $i\in[d]$, we have $\sum_{u\in V} x_{i,u} = a_i$, thus $|A_i \mathbf{z}|=a_i$. Also, $x_{i,u}\geq 0$, so for each column $k$ of $A$, we have $|\sum_{j=1}^k A_{ij} z_j|\leq a_i$. Moreover, for each $u\in V$, because of \eqref{eq:V}, we have
\[
\left|\sum_{j=1}^k A_{uj}z_j\right| \leq \max\left(\sum_{i = 1}^{d} s_i x_{i,u}+\sum_{\stackrel{uv\in E}{v<u}} y_{uv} , \sum_{\stackrel{uv\in E}{u<v}} y_{uv}\right) \leq \Delta. 
\] 
Therefore, $\Gamma \leq \max(a,\Delta)$ and we are done. 
\end{proof}

Theorems~\ref{thm:W1} and \ref{TreesW1} guarantee that \sdp is W[1]-hard with respect to the parameters $\vc$ and $d$ separately. In the following, we prove that the problem happens to be FPT with respect to the combined parameter $(\vc,d)$. First, it should be noted that ILP1 and Theorem~\ref{thm:tdILP} is not sufficient to prove FPT-ness of the problem with respect to  $(\vc,d)$ because if $A$ is the coefficient matrix of ILP1, then  $\|A\|_\infty=s=\max_i s_i$ and $s$ could be much larger than any function of $\vc$ and $d$ (in fact, vertices in the vertex cover might have large degrees and can receive stars of length much larger than any function of $\vc$ and $d$). On the other hand, the number of variables in ILP1 is $dn+m$ which is not bounded by the parameters and so Theorem~\ref{thm:ILPvar} is not applicable.  Hence, here we are going to write a new ILP for the problem, such that the number of variables is bounded by a function of $\vc$ and $d$ and thus, Theorem~\ref{thm:ILPvar} implies the assertion. 

\begin{theorem} \label{FPT:vc,d}
\sdp is FPT with respect to the parameter $(\vc,d)$. In particular, it can be solved in time at most $O^*(2^{O(d.\log d.\vc.8^{\vc})})$.
\end{theorem}

\begin{proof}
First, let $I=\Gsa$ be an instance of \sdp and $C\subset V$ be a vertex cover of size $\vc$. For convenience, set $C=[\vc]$. Also, let $ S=V\setminus C $ and $d'=\max\{i\in [d]: s_i\leq \vc\}$. Note that for each vertex $v\in S$, the degree of $v$ is at most $\vc$ and cannot be the center of any star of length  $s_i$, $i>d'$. For each $i\in [d]$ and each $u\in C$, define the variable $x_{i,u}$ which stands for the number of stars of length $s_i$ with center $u$. Also, let $\tau$ be an orientation on $E(G)$ which induces directed graph $G_\tau$. For each $e=uv\in E$ such that $u,v\in C$ and $u<v$, define the variable $y_e$ which is equal to one if $e$ is oriented in $\tau$ from $u$ to $v$ and is equal to zero, otherwise. For each $t\in [0,\vc]$, define $\mathcal{B}_t=\{\mathbf{b} \in [0,\vc]^{d'}: \sum_{i=1}^{d'} b_is_i=t \} $.
Now, for each subset $T\subseteq C$ and each vector 
$\mathbf{b} \in \mathcal{B}_{|T|}$, define a variable $z_{_{\mathbf{b},T}}$ which stands for the number of vertices $v\in S$ where $N^+(v)=T$ and for each $i\in [d']$, there are exactly $b_i$ stars of  length $s_i$ on $v$. Also, for each $T\subseteq C$, $i\in [d']$, $b\in [0,\vc]$, define
\begin{align}
z_{_T}&= \sum_{\mathbf{b}\in \mathcal{B}_{|T|}} z_{_{\mathbf{b},T}},\label{eq:zT}\\
z_{_{b,i}}&=\sum_{T\subseteq C}\ \sum_{{\mathbf{b}\in \mathcal{B}_{|T|}}:\, {b_i=b}} z_{_{\mathbf{b},T}}. \label{eq:zbi}
\end{align} 

Finally, for each $T\subseteq C$, define $N_T=\{v\in S:\, T\subseteq N(v) \}$ and for each $u\in C$, define $d_u=|E(u,S)|$. Now, we are ready to introduce the following ILP which has a feasible solution if and only if $I$ is a yes-instance. 

\begin{align} {\rm ILP2:} & \nonumber\\
& \sum_{T\subseteq C} z_{_T}= |S| & \label{eq:lp2-1} \\
&\sum_{u\in C} x_{i,u}+\sum_{b=0}^{\vc} b\, z_{_{b,i}} = a_i & \forall\, i \in [d'] \label{eq:lp2-2} \\
&\sum_{u\in C} x_{i,u} = a_i & \forall\, i \in [d'+1,d] \label{eq:lp2-3} \\
&\sum_{i = 1}^{d} s_i x_{i,u} = \sum_{\stackrel{uv\in E(C)}{u<v}} y_{uv}+\sum_{\stackrel{uv\in E(C)}{v<u}} (1-y_{vu})+d_u- \sum_{\stackrel{T\subseteq C}{ u\in T}} z_{_T}  & \forall\, u \in C \label{eq:lp2-4} \\
&\sum_{T\in \mathcal{T}} z_{_T}\leq  |\bigcup_{T\in \mathcal{T}} N_T| & \forall \mathcal{T}\subseteq  2^{C} \label{eq:lp2-5} \\
&x_{i,u} \in \mathbb{Z}^{\geq 0} &\forall\, i \in [d], \forall u\in C \nonumber\\
&y_{uv} \in \{0,1\}  & \forall\, uv\in E(C) \nonumber \\
&z_{_{\mathbf{b},T}} \in \mathbb{Z}^{\geq 0}  & \forall\, T\subseteq C, \mathbf{b} \in \mathcal{B}_{|T|}\nonumber 
\end{align}

Note that $z_{_T}$ is the number of vertices $v$ in $S$ with $N^+(v)=T$, so Constraint \eqref{eq:lp2-1} assures that every vertex in $S$ is enumerated in some $z_{_T}$, $T\subseteq C$. Constraints~\eqref{eq:lp2-2} and \eqref{eq:lp2-3} guarantee that the total number of used stars of length $i$ is exactly equal to $a_i$, for each $i\in [d]$ (note that vertices in $S$ can receive only stars with length $s_i$, $i\in [d']$, and vertices in $C$ can receive stars of all lengths $s_i$, $i\in [d]$). Constraints~\eqref{eq:lp2-4} show that the sum of lengths of stars on each vertex $u\in C$ is equal to $N^+(u)$. Finally, Constraints~\eqref{eq:lp2-5} hold because if a vertex $v\in S$ is enumerated in some $z_{_T}$, then $v\in N_{T}$. Now, we prove that $I$ is a yes-instance if and only if ILP2 has a feasible solution. 

\textbf{Proof of Necessity}. Suppose that $I$ is a yes-instance and consider an \sasd of $G$. Also, let $G_\tau$ be its corresponding orientation on $G$. Assign to each of the variables $x_{i,u}$, $y_{uv}$ and $z_{_{\mathbf{b},T}}$ its corresponding value according to their definitions and define $z_{b,i}$ and $z_{_T}$ as in \eqref{eq:zT} and \eqref{eq:zbi}. Also, for each $T\subseteq S$, define $Z_T=\{v\in S:\ N^+(v)=T \}$. It is clear that all $Z_T$'s are disjoint, $|Z_T|=z_{_T}$ and each vertex $v\in S$ is in some $Z_T$, therefore, Constraint~\eqref{eq:lp2-1} holds.  
Now, for each star size $s_i$, $i\in [d']$, there are exactly $x_{i,u}$ stars of length $s_i$ on each vertex $u\in C$ and the number of stars of length $s_i$ on the vertices in $S$ is equal to $\sum_{b=0}^{\vc} b\, z_{_{b,i}} $ (because each vertex $v\in S$ with exactly $b$ stars of length $s_i$, for some $b$, is enumerated in $z_{_{b,i}}$). Thus, Constraints~\eqref{eq:lp2-2} hold. On the other hand, for each star length $s_i$, $i\in [d'+1,d]$, no vertices in $S$ can have a star of length $s_i$, so Constraints~\eqref{eq:lp2-3} also hold. Now, note that for each $u\in C$, $\sum_{i=1}^d s_i x_{i,u}$ is equal to $|N^+(u)|=|N^+_S(u)|+|N^+_C(u)|$. Also, it is clear that $|N^+_C(u)|= \sum_{{uv\in E(C)},\, {u<v}} y_{uv}+\sum_{{uv\in E(C)},\, {v<u}} (1-y_{uv}) $ and $|N^+_S(u)|= d_u - \sum_{{T\subseteq C}, { u\in T}} z_{_T}$ (since $N^-_S(u)= \cup_{ T\ni u} Z_T$ and $Z_T$'s are disjoint). Finally, for each $T\subseteq C$, $Z_T\subseteq N_T$ and $Z_T$'s are disjoint, so Constraints~\eqref{eq:lp2-5} hold. This implies that there exists a feasible solution to ILP2, as desired.

\textbf{Proof of sufficiency}. Now, suppose that $(x_{i,u},y_{uv},z_{_{\mathbf{b},T}})$ is a feasible solution to ILP2 and for each $T\subseteq C$, define $z_{_T}$ as in \eqref{eq:zT}. First, we define an orientation $\tau$ on $G$ as follows. If $e=uv$ is an edge with both endpoints in $C$ where $u<v$, then orient $e$ from $u$ to $v$ if and only if $y_{uv}=1$. Now, apply Theorem~\ref{thm:SDR} by setting $\mathcal{F}=(N_T: T\subseteq C)$ and $\eta(T)=z_{_T}$, for each $T\subseteq C$. Since Constraints~\eqref{eq:lp2-5} hold, there exists an $\eta$-SDR for $\mathcal{F}$, i.e. for each $T\subseteq C$, there exists a subset $Z_T\subseteq N_T$ such that $|Z_T|=z_{_T}$ and $Z_T$'s are disjoint. Note that because of Constraint~\ref{eq:lp2-1}, the sets $Z_T$, $T\subseteq C$, partition the set $S$. Now, for each edge $e=uv$, where $u\in C$ and $v\in Z_T$, for some $T\subseteq C$, orient $e$ from $v$ to $u$ if $u\in T$ and from $u$ to $v$ if $u\in C\setminus T$. This gives rise to an orientation $\tau$ on $G$. Now, we give an \sasd for $G_\tau$. First, because of \eqref{eq:zT}, the set $Z_T$ can be partitioned into arbitrary disjoint subsets $Z_{\mathbf{b},T}$, $\mathbf{b}\in \mathcal{B}_{|T|}$, where $|Z_{\mathbf{b},T}|=z_{_{\mathbf{b},T}}$. Now, for each $i\in [d]$ and $u\in C$, assign $x_{u,i}$ stars of length $s_i$ to the vertex $u$. Also, for each $i\in [d']$ and $v\in S$, if $v\in Z_{\mathbf{b},T}$, for some $T\subseteq C$ and $\mathbf{b}\in \mathcal{B}_{|T|}$, then assign $b_i$ stars of length $s_i$ to the vertex $v$. Since for each $v\in Z_{\mathbf{b},T} $, $|N^+(v)|=|T|=\sum_{i=1}^{d'} b_is_i$,  the outdegree of $v$ is equal to the sum of lengths of its assigned stars. On the other hand, for each $u\in C$, $|N^+(u)|=  \sum_{{u<v}} y_{uv}+\sum_{{v<u}} (1-y_{uv})+d_u- \sum_{{T\subseteq C},{ u\in T}} z_{_T}$ and so Constraints~\eqref{eq:lp2-4} guarantee that the outdegree of $u$ is equal to the sum of lengths of its assigned stars. Finally, because of Constraints~\eqref{eq:lp2-2} and \eqref{eq:lp2-3} all stars have been assigned. Hence, $G$ admits an \sasd.

\textbf{Time analysis.} ILP2 has $d.\vc$ variables $x_{i,u}$, at most $\vc^2$ variables $y_{uv}$, $\sum_{t=0}^{\vc} \binom{\vc}{t} |\mathcal{B}_t|$ variables $z_{_{\mathbf{b},T}}$. Therefore, the total number of variables is at most
\[
d.\vc+\vc^2+ \sum_{t=0}^{\vc} \binom{\vc}{t} \binom{t+d'}{d'}\leq  d.\vc+\vc^2+ 2^{3\vc}=O(d.2^{3\vc}).
\]
Hence, by Theorem~\ref{thm:ILPvar}, ILP2 can be solved in $2^{O(d.\log d.\vc. 2^{3\vc})}$. Also, since Constraints \ref{eq:lp2-5} hold, applying Theorem~\ref{thm:SDR}, we can find an $\eta$-SDR for $\mathcal{F}=(N_T: T\subseteq C)$ and $\eta(T)=z_{_T}$ in time $O(|S|^3)$. Using this SDR along with variables $y_{uv}$ we can introduce an orientation $\tau$ on $G$. Also, according to the above argument in the proof of sufficiency, assigning the stars to the vertices of $G$  using variables $x_{i,u}$ and $z_{_{\mathbf{b},T}}$ can be done in polynomial time. Hence, the whole process can be done in $O^*(2^{O(d.\log d.\vc. 2^{3\vc})})$. 
\end{proof}

\begin{remark}\label{rem:multi}
We can also define the \sdp problem for multigraphs and multi-stars.
Let $G=(V,E)$ be a loopless multigraph. Also, let $(\mathbf{s},\mathbf{a})\in \mathbb{Z}_+^d\times \mathbb{Z}_+^d$. By a multi-star of length $s_i$, we mean a bipartite multigraph whose one part is of size one and its number of edges is equal to $s_i$. So, the \msdp problem asks whether there exists a decomposition of $G$ into $a_1$ multi-stars of length $s_1$, $\cdots$, $a_d$ multi-stars of length $s_d$.
Note that if $G=(V,E)$ is a multigraph with vertex cover $C$ such that all edges in $E(C,V\setminus C)$ are simple (i.e. all parallel edges are inside $C$), then \msdp\ can be solved for $G$ with the same runtime as in   Theorem~\ref{FPT:vc,d}, because in ILP2, we can replace the constraint $y_{uv}\in \{0,1\}$ with $y_{uv}\in [0,m_{uv}]$, where $m_{uv}$ is the multiplicity of the edge $uv$ and everything works accordingly. We will need this fact in the proof of Theorem~\ref{thm:s,nd}.
\end{remark}

In previous result, we proved that \sdp is FPT with respect to $(\vc,d)$ and in Theorem~\ref{thm:W1} and \ref{TreesW1}, we proved that \sdp is W[1]-hard with respect to $\vc$. In the following theorem, we will prove that the problem is in XP with respect to $\vc$.

\begin{theorem}\label{XPvc}
\sdp is in \XP with respect to the parameter $\vc$. More precisely, if a vertex cover of size $\vc$ is given, then \sdp can be solved in time  $O({(n+5^{\vc})}^{5^{\vc}} n^{\vc})$, where $n$ is the number of vertices of the input graph. 
\end{theorem}
\begin{proof}
Let $G=(V,E)$ and $(\mathbf{a},\mathbf{s})$ be an instance of \sdp. Also, let $C=\{u_1,\ldots, u_{\vc}\}\subset V$ be a minimum vertex cover of size $\vc$ and $I=V\setminus C$ (note that $C$ can be found in \FPT time). Let $d'$ be the maximum $i$ such that $s_i\leq \vc$ and for each $A\subseteq C$, define
\begin{align*}
I_A&= \{x\in I:\ N(x)=A \},\\
\CC_A&= \{(\mathbf{b},B):\ \mathbf{b}=(b_1,\ldots,b_{d'}), \forall i\in[d']\ b_i\leq a_i, B\subseteq A \text{ and } \sum_{i=1}^{d'} s_i b_i= |B|
\}.
\end{align*}
To each vertex $x\in I_A$, we assign a label $(\mathbf{b},B)$ in $\mathcal{C}_A$ which means that there are exactly $b_i$ stars of length $s_i$ with the center $x$, $i\in[d']$, and $B$ consists of all the leaves of these stars. For a fixed $A\subseteq C$, consider a solution of the following equation
\begin{equation}
\label{eq:eqIA}
\sum_{(\mathbf{b},B)\in \mathcal{C}_A} \alpha_{(\mathbf{b},B)} =|I_A|.
\end{equation}
This means that there are exactly $ \alpha_{(\mathbf{b},B)} $ vertices in $I_A$ with the label $(\mathbf{b},B)$. Now, for each $A\subseteq C$, fix a solution of \eqref{eq:eqIA} and assign stars with the center $x$ for each $x\in I_A$ according to its label $(\mathbf{b},B)$. Also, let $G'$  be the graph obtained from $G$ by removing edges of all assigned stars (i.e. all edges in $E(x,B )$) and $\mathbf{a'} $ be the list obtained from $\mathbf{a}$ by subtracting the multiplicities of the assigned stars (i.e. $a'_i=a_i-b_i$ for $i\in [d']$ and $a'_i=a_i$ for $i\in [d'+1,d]$). Now, fix an orientation $\tau$ on the edges of $G'$ with both endpoints in $C$ and orient all the remaining edges in $G'$ from $C$ to $I$ and let $G'_\tau$ be the directed graph obtained from $G'$ with such orientation. Now, in the final step we check that if $G'_\tau$ admits an $(\mathbf{a}',\mathbf{s})$ star decomposition.
This can be done by a simple dynamic programming as follows.

For each $i\in [\vc]$, let $d_i$ be the out-degree of the vertex $u_i$ in $G'_\tau$. Therefore, $\mathbf{a}'.\mathbf{s}= \sum_{i=1}^{\vc} d_i$. We are going to distribute the remaining stars over the vertices in $C$ such that the sum of star lengths assigned to each vertex is is equal to its out-degree.  

To do this, let $T=\mathbf{a}'.\mathbf{1}$ be the number of remaining stars and let $\eta: [T]\to \{s_1,\ldots, s_d\} $ be a function such that $|\eta^{-1}(s_i)|= a'_i$, for each $i\in [d]$, and  for every $i \in \{1,\ldots,T\}$ and  integers $ x_j \in [0,d_j]$, $ j\in [ \vc-1]$, define

\[
f(i,x_1,\ldots, x_{\vc-1})=
\begin{cases}
1 & \text{if }\exists I_1,\ldots, I_{\vc-1}\subseteq \{1,\ldots,i\}\  s.t.\  I_k\cap I_l =\emptyset, k\neq l, \\
& \text{and } \sum_{r\in I_j} \eta(r) =x_j,  j\in [\vc-1], \\ 
0 & \text{otherwise.}
\end{cases}
\]
It is clear that $G'_\tau$ admits an $(\mathbf{a}',\mathbf{s})$ star decomposition if and only if $f(T,d_1,\ldots,d_{\vc-1})=1$. Now, we compute the function $f$ recursively using the following dynamic programming. 

First, note that $f(1,x_1,\ldots, x_{\vc-1})=1$ if and only if all $x_j$'s are equal to zero except at most one $x_{j_{_0}}$, where $x_{j_{_0}}=\eta(1)$. 
Moreover, note that if $f(i,x_1,\ldots, x_{\vc-1})=1$, then clearly we have  $f(i+1,x_1,\ldots ,x_{\vc-1})=1$. Also, if there is some $j\in [\vc-1] $
such that $x_j\geq \eta(i+1)$ and $f(i,x_1,\ldots,x_{j-1}, x_j-\eta({i+1}),x_{j+1},\ldots, x_{\vc-1})=1$, then we can replace $I_j$ with $I_j\cup \{i+1\}$ and so again $f(i+1,x_1,\ldots ,x_{\vc-1})=1$. If none of the above two cases happen, then clearly $f(i+1,x_1,\ldots ,x_{\vc-1})=0$. Therefore, 
\begin{align*}
f(i+1,x_1,\ldots ,x_{\vc-1})=1 \text{ iff} & \text{ either } f(i,x_1,\ldots, x_{\vc-1})=1, \text{ or } \exists\, j\in[\vc-1], \\
& f(i,x_1,\ldots,x_{j-1}, x_j-\eta({i+1}),x_{j+1},\ldots, x_{\vc-1})=1.
\end{align*}

Hence, we can compute the function $f$ recursively from $i=1$ to $i=T$. It takes $O(\vc)$ to compute each cell of the table. So, the running time of the dynamic programming is $O(T.n^{\vc-1}.\vc)= O(\vc.n^{\vc+1})$, since $T\leq n^2$.

Now, we analyze the time complexity of the algorithm. For each $A\subseteq C$ and $B\subseteq A$ with $|B|=k$, the number of labels $(\mathbf{b},B)$ in $\mathcal{C}_A$ is at most $(k+1)(k/2+1)\cdots (k/k+1)={(2k)!}/{k!^2}$ (because $0\leq b_i\leq k/i$). Therefore, 
\begin{equation*}
|\mathcal{C}_A|\leq \sum_{k=0}^{|A|}\binom{2k}{k}\binom{|A|}{k}\leq \sum_{k=0}^{|A|} 2^{2k}\binom{|A|}{k}=5^{|A|}\leq 5^{\vc}.
\end{equation*}
So, the number of  solutions to Equation \eqref{eq:eqIA} is at most $ \binom{n+5^{\vc}-1}{5^{\vc}	-1}\leq \frac{ {(n+5^{\vc})}^{5^{\vc}-1}}{(5^{\vc}-1)!} $.
Since we take a solution of \eqref{eq:eqIA} for each $A\subseteq C$ and take an orientation $\tau$ (which has at most $2^{\vc^2}$ choices) and then solve a dynamic programming, the running time of the whole algorithm is at most 

\[\frac{ {(n+5^{\vc})}^{5^{\vc}-1}}{(5^{\vc}-1)!}. 2^{\vc^2}.2^{\vc}. O(\vc. n^{\vc+1})\leq O({(n+5^{\vc})}^{5^{\vc}} n^{\vc}). \]
%
\end{proof}

\section[Fixed parameter tractability with respect to (nd,s)]{Fixed parameter tractability with respect to $(\nd,s)$}\label{sec:nd,s}
In this section, we prove that \sdp is FPT with respect to $(\nd,s)$.
\begin{theorem} \label{thm:s,nd}
	\sdp is FPT with respect to the parameter $(\nd,s)$.
\end{theorem}

In order to prove Theorem~\ref{thm:s,nd}, we need a couple of lemmas. Let $I=\Gsa$ be an instance of \sdp. First, recall from \eqref{eq:expansion} the definition of the edge expansion of $G$ denoted by $\varphi(G)$. Theorem~\ref{thm:tarsi} asserts that if $s\leq \varphi(G)$, then $I$ is a yes-instance. Here, we prove a generalization of this result, which guarantees the existence of star decomposition even when $s\leq \varphi(G\setminus S)$ when $S$ is a stable set in $G$.  


\begin{lemma} \label{lem:phi}
Let $I=\Gsa$ be an instance of \sdp. Also, let $S$ be a stable set in $G$. If $s=\max_i{s_i} \leq \varphi(G\setminus S)$, then $I$ is a yes instance. In particular, $G$ admits an $\sa$-star decomposition such that the center of all stars are in $V(G)\setminus S$.
\end{lemma}

In order to prove Lemma~\ref{lem:phi}, we need the following result from \cite{tarsi} about existence of a orientation of a graph with prescribed out-degrees.

\begin{lemma}\label{lem:orient} {\rm \cite{tarsi}}
	Let $H$ be a graph on the vertex set $V=[n]$ and $\mathbf{d}^+=(d^+_1,\ldots,d^+_n)$ be a sequence of nonnegative integers.
	For each $i\in [n]$, define $\delta(i)=\deg(i)-2d^+_i$.
	There exists an orientation of $H$ such that $\deg^+(i)=d^+_i$ for each $i\in [n]$ if and only if
	\begin{enumerate}
		\item $\delta(V)=0$,
		\item for each $A\subseteq V$, $\delta(A)\leq |E(A,V\setminus A)|$,
	\end{enumerate}  
 where $\delta(A)=\sum_{i\in A} \delta(i)$.
\end{lemma}

\begin{proof}[Proof of Lemma~\ref{lem:phi}]
Let $\mathcal{S}$ be a multiset consisting of integer $s_i$ with multiplicity $a_i$, for each $i\in [d]$. Also, let $V\setminus S= [n]$. Now, partition $\mathcal{S}$ into $n$ multisets $\mathcal{B}=(B_1,\ldots, B_n)$ and define $d'_i= \sum_{x\in B_i} x$. 
Moreover, for each $i\in [n]$, define $\deg'(i)=\deg_G(i)+|E(i,S)|$ (note that any edge from $i$ to $S$ is counted twice in $\deg'(i)$) and define $\delta'(i)=\deg'(i)-2d'_i$. Also, for each $A\subseteq [n]$, define $\delta'(A)=\sum_{i\in A} \delta'(i)$. We choose $\mathcal{B}$ such that 
\begin{equation}\label{eq:equitable}
\max_i \delta'(i)-\min_i \delta'(i)\leq 2s.
\end{equation}
 This can be done, because if there are two $i,j\in[n]$ such that $\delta'(i)-\delta'(j)>2s$, then we can move an arbitrary integer $s_k$ from $B_j$ to $B_i$ and reduce $|\delta'(i)-\delta'(j)|$ and doing this iteratively we can lead to a partition satisfying Condition \eqref{eq:equitable}.

 First, it is clear that $\delta'([n])=0$. Now, we prove that for each $A\subseteq [n]$, $\delta'(A)\leq d(A,[n]\setminus A)$. To see this, take a subset $A\subseteq [n]$ and let $r=|A|$ and $\mu=\delta'(A)/r$. By \eqref{eq:equitable}, we have $\delta'(i)\geq \mu-2s$, for each $i\in [n]$. Therefore,
 
 \begin{align*}
 -\mu r= -\delta'(A)=\delta'([n]\setminus A) \geq (\mu-2s)(n-r).
 \end{align*}   

So, $\mu n\leq 2s(n-r)$ and thus,
\[
\delta'(A)= \mu r \leq 2s\, \dfrac{r(n-r)}{n}\leq 2\varphi(G\setminus S)\, \dfrac{r(n-r)}{n} \leq d(A,[n]\setminus A).
\]
Now, for each $i\in [n]$, let $d^+_i=d'_i-|E(i,S)|$ and $\delta(i)=\deg_{G\setminus S}(i)-2d^+_i$. It is clear that for each $i\in [n]$, $\delta'(i)=\delta(i)$. So, conditions of Lemma~\ref{lem:orient} hold for $H=G\setminus S$ and thus, there exists an orientation of $G\setminus S$ such that $\deg^+_{G\setminus S}(i)=d^+_i$ for all $i\in [n]$. Finally, we orient all edges $uv$, $u\in [n], v\in S$, from $u$ to $v$. Hence, for each $i\in [n]$, $\deg^+_G(i)=d'_i$ and we assign all stars with lengths in $B_i$ to the center $i$ and obtain an \sasd\ for $G$. 
\end{proof}

We also need the following observation regarding the edge expansion of complete and complete bipartite graphs. 
\begin{lemma} \label{lem:CBphi}
Let $G=(V, E)$ be a complete bipartite graph with parts $X$ and $Y$ such that $|X|,|Y|\geq n$. Then, we have $\varphi(G)\geq n/4$. Also, let $G=(V,E)$ be complete graph on at least $n$ vertices. Then, we have $\varphi(G)\geq n/2$.
\end{lemma}
\begin{proof}
Suppose that $G=(V, E)$ be the complete bipartite graph with parts $X$ and $Y$, where $|X|=n_1\geq n$ and $|Y|=n_2\geq n$ and let $(S,V\setminus S)$ be the bipartition achieving $\varphi(G)$, and without loss of generality assume that $|S\cap X|\leq n_1/2$. Now, if $|S\cap Y|\geq n_2/2$, then, $d(S,V\setminus S)\geq n_1n_2/4$ and we have
\[
\varphi(G)= \dfrac{1}{2}\left(\dfrac{1}{|S|}+\dfrac{1}{|V\setminus S|}\right)d(S,V\setminus S)\geq \dfrac{1}{8}\dfrac{n_1n_2(n_1+n_2)}{|S||V\setminus S|}\geq \dfrac{1}{2}\dfrac{n_1n_2}{n_1+n_2}\geq \dfrac{n}{4}.
 \]
Now, suppose that $|S\cap Y|\leq n_2/2$. Then, 
\[
\varphi(G)\geq \dfrac{1}{2}\left(\dfrac{d(S,V\setminus S)}{|S|}\right)\geq \dfrac{|S\cap X|\dfrac{n_2}{2}+|S\cap Y|\dfrac{n_1}{2}}{2|S|}\geq \dfrac{n|S|}{4|S|}=\dfrac{n}{4}.
\]
Thus, we always have $\varphi(G)\geq{n}/{4}$, as desired.

Finally, let $G=(V,E)$ be the complete graph on at least $n$ vertices and $(S,V\setminus S)$ be a partition of $V$. Then, $d(S,V\setminus S)=|S||V\setminus S|$ and thus 
\[\varphi(G)= \dfrac{1}{2}\left(\dfrac{1}{|S|}+\dfrac{1}{|V\setminus S|}\right)d(S,V\setminus S)\geq \dfrac{1}{2}\dfrac{|S||V\setminus S| n}{|S||V\setminus S|}=\dfrac{n}{2}.
\]
\end{proof}

Now, we are ready to prove Theorem~\ref{thm:s,nd}. 
\begin{proof}[Proof of Theorem~\ref{thm:s,nd}.]
	Let $I=\Gsa$ be an instance of \sdp such that $\nd(G)={\nd}$ and $s=\max_i s_i$. Also, let $(V_1,\ldots,V_{\nd})$ be a neighborhood diversity decomposition, where $|V_1|\geq |V_2|\geq \cdots\geq |V_{\nd}|$. 
	
	First, note that if $s\geq \sqrt{\log\log n}/4$, then $n\leq 2^{2^{16s^2}}$ and ILP1 has at most $2^{2^{16s^2+1}}+s2^{2^{16s^2}}$ variables. Therefore, by Theorem~\ref{thm:ILPvar}, the problem can be solved in FPT time with respect to $s$. Thus, set $m=\sqrt{\log\log n}$ and suppose that $s\leq  m/4$.
	
	Now, suppose that $k\in[nd]$ be such that $|V_k|\geq m>|V_{k+1}|$ (note that if $|V_{nd}|\geq m$, then we take $V_{nd+1}$ as the empty set). Now, let $\hat{V}=V_1\cup\cdots\cup V_k$. We construct a partition of $\hat{V}$ as follows. 

First, let $\hat{G}$ be the graph whose vertices are corresponding to $V_1,\dots,V_k$ and two vertices corresponding to $V_i$ and $V_j$ are adjacent if $V_i$ is complete to $V_j$. Now, set $i=1$ and do the following procedure. 
 
	\begin{enumerate}
		\item Consider a non-null connected component of $\hat{G}$, let $T$ be a spanning tree of this component rooted at a vertex $r$.
		\item Consider a vertex $v\in V(T)$ whose all children are leaves and  set $B_i$ be the union of all sets $V_j$ corresponding to the vertex $v$ and all its children and remove all these vertices from $T$. Now, if $V(T)=\{r\}$, then add the set $V_j$ corresponding to $r$ to $B_i$. If $V(T)=\emptyset$, then go to (3), otherwise set $i=i+1$ and repeat (2).
        \item Repeat (1) and (2) for all non-null connected components of $\hat{G}$. If all remaining vertices in $\hat{G}$ are singletone vertices, then go (4).
		\item If there is some $v\in V(\hat{G})$ such that its corresponding set $V_j$ is a clique, then remove $v$ from $\hat{G}$ and set $B_i=V_j$,   $i=i+1$ and repeat (4). Otherwise, set $B_0$ be the union of all sets $V_j$ corresponding to remaining vertices in $\hat{G}$ and halt. 
	\end{enumerate}

    \begin{figure}{ht}\
\includegraphics[width=\textwidth,
height=\textheight,
keepaspectratio,
trim=50 80 50 50pt]{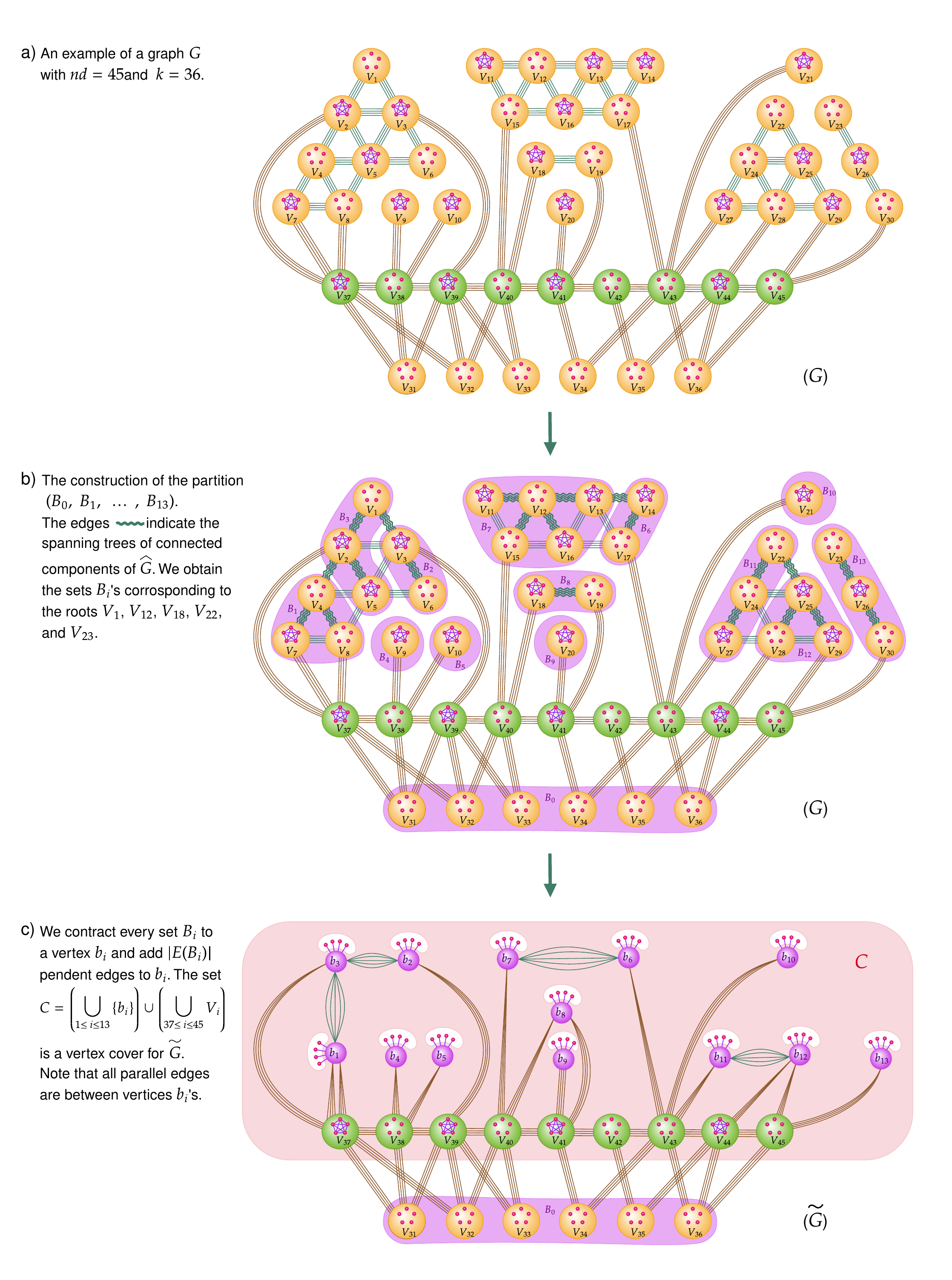}
\vspace{10pt}\caption{An example of a graph $G$ with its neighborhood diversity decomposition and the construction of the graph $\tilde{G}$ from $G$.}\label{fig1}
\end{figure}

Let $(B_0,B_1,\ldots, B_{t})$ be the partition of $\hat{V}$ obtained as above. It is clear that the induced graph $G[B_0]$ is an stable set and there is no edge between $B_0$ and $\cup_{i=1}^t B_i$. Also, $G[B_i]$ for each $i\geq 1$, either is a complete graph, or contains a complete bipartite graph as a spanning subgraph. Now, we construct the multigraph $\tilde{G}$ from $G$ as follows. First, for each $i\in[t]$ contract all vertices in $B_i$ into a single vertex $b_i$ and add $|E(B_i)|$ pendent edges to $b_i$. Note that $|E(\tilde{G})|=|E(G)|$ (for an example, see the figure \ref{fig1}). First, we prove that $I $ is yes instance for \sdp\ if and only if $\tilde{I} $ is a yes instance for \msdp.  

First, note that if $G$ admits an $\sa$-star decomposition, then by contracting $G$ to $\tilde{G}$ and the star decomposition to a multi-star decomposition, we can clearly see that $\tilde{G}$ admits a multi-star decomposition. Now, assume that $\tilde{I}$ is a yes-instance for \msdp. The solution for $\tilde{I}$ induces an orientation on the edges of $\tilde{G}$. For each $i\in [t]$, let $G_i$ be the subgraph of $G$ obtained from the induced subgraph $G[B_i]$ by adding all edges which are directed from $B_i$ to $V\setminus B_i$ (corresponding to the orientation of $\tilde{G}$). Also, let $(\mathbf{s}_i,\mathbf{a}_i)$ be the star vectors assigned to $b_i$. By Lemma~\ref{lem:CBphi}, we know that $\varphi(G[B_i])\geq m/4\geq s$. So, by Lemma~\ref{lem:phi}, $G_i$ admits an $(\mathbf{s}_i,\mathbf{a}_i)$-star decomposition. Finally, in the graph $\tilde{G}$ blow-up all vertices $b_i$ into $B_i$  to recover the graph $G$ and in the multi-star decomposition of $\tilde{G}$, replace all stars on vertices $b_1,\ldots, b_t$ with star decompositions of $G_1,\ldots, G_t$ and keep other stars assigned to vertices in $V(G)\setminus \cup_{i=1}^t B_i$ unchanged. This gives an star decomposition of $G$. Therefore, $I$ is a yes-instance, as desired.

Now, we prove that \msdp can be solved for $\tilde{I}$ in FPT-time with respect to $(\nd,d)$. To see this, first note that $C=V_{k+1}\cup \cdots\cup V_{\nd}\cup \{b_1,\ldots,b_t\}$ is a vertex cover for $\tilde{G}$ and $\vc(\tilde{G})\leq (\nd-k)m +t\leq m \nd $. Also, there is no edge between $B_0$ and $\{b_1,\ldots, b_t\}$, so every edge in $E(C,V(\tilde{G}\setminus C))$ is simple. Thus, by Remark~\ref{rem:multi}, Theorem~\ref{FPT:vc,d} implies that the \msdp\ can be solved for the instance $\tilde{I}=(\tilde{G},\mathbf{s},\mathbf{a}) $  in time  $O^*(2^{O(d.\log d.m.\nd.2^{3m\nd})})$. Let $k=\nd d\log d$ and note that the runtime is at most $O^*(2^{O(k.m.2^{3km})})$. First, if $k\geq m/6$, then $n\leq 2^{2^{36k^2}}$ and by Theorem~\ref{thm:ILPvar}, ILP1 can be solved in FPT-time with respect to $k$ and so $(\nd,d)$. So, suppose that $k\leq m/6$. In this case, the runtime is at most $O^*(2^{O(m^22^{m^2/2})})=O^*(2^{O(\log \log n \sqrt{\log n})})\leq O^*(2^{o(\log n)}) $ which is a polynomial of $n$. 
\end{proof}

\section{Concluding Remarks}
In this paper, we studied parameterized complexity of the star decomposition problem. We proved
that the problem is para-NP-hard with respect to the tree-depth, the neighborhood diversity and the
feedback vertex number of the input graph as well as the number of star types $d$ and the
maximum size of the stars $s$. It is also W[1]-hard and in XP with respect to the vertex cover
number of the input graph. On the other hand, the problem is FPT with respect to combined
parameters $(\vc,d)$, $(\td,s)$ and $(\nd,s)$. Moreover, it is W[1]-hard with respect to $(\td,d)$,
$(\nd,d)$ and $(\fvs,d)$. The presented parameterized landscape gives rise to the natural question
that if the problem is FPT with respect to $(\tw,s)$ or $(\fvs,s)$? 
Even we do not know if the problem
is FPT with respect to $s$ on trees (or forests). When the input graph is a star forest, then the
problem is equivalent to the bin packing problem which is known to be FPT with respect to the
maximum size of items (which is equivalent to $s$) \cite{koutecky}. Thus, the first step is to prove
that the star decomposition problem is FPT with respect to $s$ on caterpillar trees. A second
question is that if the problem is in XP with respect to $(\nd,d)$?

\bibliographystyle{unsrt}
\bibliography{References}
\end{document}